\documentclass[11pt]{article}
\usepackage[letterpaper,margin=1in]{geometry}
\usepackage{fullpage}
\usepackage{amsmath, amsthm, amssymb}

\usepackage{mathrsfs}

\usepackage{times}

\newtheorem{theorem}{Theorem}[section]
\newtheorem{lemma}[theorem]{Lemma}

\newcommand{\chan}{{\cal C}}
\newcommand{\R}{\mathscr{R}}

\renewcommand{\H}{\mathscr{H}}

\begin{document}

\title{Radio Network Lower Bounds Made Easy}

\author{Calvin Newport\\ Georgetown University\\ {\tt cnewport@cs.georgetown.edu}}

\date{}

\maketitle

\begin{abstract}
Theoreticians have studied distributed algorithms in the radio network model for close to three decades.
A significant fraction of this work focuses on lower bounds for basic communication problems
such as {\em wake-up} (symmetry breaking among an unknown set of nodes) and {\em broadcast}
(message dissemination through an unknown network topology).
In this paper, we introduce a new technique for proving this type of bound, based on reduction from a probabilistic hitting game,
 that simplifies and strengthens much of this existing work. 
In more detail, in this single paper we prove new expected time and high probability
lower bounds for wake-up and global broadcast in single and multi-channel versions
of the radio network model both with and without collision detection.
In doing so, we are able to reproduce results that previously spanned a half-dozen papers
published over a period of twenty-five years. In addition to simplifying these existing results, our technique, in many places,
also improves the state of the art: of the eight bounds we prove, four strictly strengthen the best known previous result
(in terms of time complexity and/or generality of the algorithm class for which it holds),
and three provide the first known non-trivial bound for the case in question.
The fact that the same technique can easily generate this diverse collection of lower bounds
indicates a surprising unity  underlying communication tasks in the radio network model---revealing that 
deep down, below the specifics of the problem definition and model assumptions,
 communication in this setting reduces to finding efficient strategies for a simple game.
  \end{abstract}

%\vspace{45mm}

%\noindent {\em Regular paper.}
%\setcounter{page}{0}
%\thispagestyle{empty}
%\newpage

%%%%%%%%%%%%%%%%%%%%%%%%%%%
%%%%%%%%%%%%%%%%%%%%%%%%%%%

\section{Introduction}
\label{sec:intro}

In this paper, we introduce a new technique for proving lower bounds for basic communication tasks
in the radio network model.
We use this technique to unify, simplify, and in many cases strengthen the best known lower bounds
for two particularly important problems: wake-up and broadcast.

\noindent {\bf The Radio Network Model.}
The radio network model represents a wireless network as a graph $G=(V,E)$, where the nodes in $V$ correspond
to the wireless devices and the edges in $E$ specify links. Each node can broadcast messages to its neighbors
in $G$. If multiple neighbors of a given node broadcast during the same round, however, the messages are lost due to collision.
This model was first introduced by Chlamtac and Kutten~\cite{chlamtac:1985}, who used it to study centralized algorithms.
Soon after, Bar-Yehuda et al.~\cite{baryehuda:1987,baryehuda:1992} introduced the model to the distributed algorithms community
where variations have since been studied in a large number of subsequent papers; 
e.g.,~\cite{alon:1991,kushilevitz:1998,jurdzinski:2002,moscibroda:2005,colton:2006,kowalski:2005,czumaj:2006,dolev:2007,dolev:2008,gilbert:2009,dolev:2009,dolev:2011,daum:2012,daum:2012b,ghaffari:2013,ghaffari:2013b}.
%The radio network model is not meant to provide a high fidelity representation of real world radio communication.
%It instead cleanly abstracts a core difficulty of this setting: managing contention with unknown participants and/or topology.

Two of the most investigated problems in the radio network model are {\em wake-up} (basic symmetry breaking among an unknown
set of participants in a single hop network) and {\em broadcast} (propagating a message from a source to all nodes in an unknown multihop network).
Lower bounds for these problems are important because wake-up and/or broadcast reduce to most useful communication tasks in this setting, and
therefore capture something fundamental about the cost of distributed computation over radio links.
%With this mind, lower bounds are of particular interest in this setting
%as they help quantify the impact of this defining property of the wireless world.

\noindent {\bf Our Results.}
In this paper, we introduce a new technique (described below) for proving lower bounds for wake-up and broadcast 
in the radio network model. 
We use this technique to prove 
 new expected time
and high probability lower bounds for these two problems in the single and multiple channel versions of
the radio network model both with and without collision detection.
In doing so, we reproduce in this single paper a set of existing results that spanned a half-dozen 
papers~\cite{willard:1986,kushilevitz:1998,jurdzinski:2002,colton:2006,dolev:2009,daum:2012b}
 published
over a period of twenty-five years. Our technique simplifies these existing arguments and establishes 
a (perhaps) surprising unity among these diverse problems and model assumptions.
Our technique also strengthens the state of the art.
All but one of the results proved in this paper improve the best known existing result by increasing
the time complexity and/or generalizing the class of algorithms for which the bound holds (many existing
bounds for these problems hold only for {\em uniform} algorithms that require nodes to use a pre-determined
sequence of independent broadcast probabilities; all of our lower bounds,
by contrast,
hold for all randomized algorithms).
In several cases, we prove the first known bound for the considered assumptions.

The full set of our results with comparisons to existing work are described in Figure~\ref{fig:results}. 
Here we briefly mention three highlights (in the following, $n$ is the network size and $D$ the network diameter).
In Section~\ref{sec:wakeup:cd}, we significantly simplify Willard's seminal $\Omega(\log\log{n})$ bound for expected time
wake-up with collision detection~\cite{willard:1986}. In addition, whereas Willard's result only holds for uniform algorithms,
our new version holds for all algorithms. In Section~\ref{sec:wakeup:multi}, we prove the first tight bound for high probability
wake-up with multiple channels and the first known expected time bound in this setting. 
And in Section~\ref{sec:bcast}, we prove that Kushilevitz and Mansour's oft-cited $\Omega(D\log{(n/D)})$ lower bound for expected time 
broadcast~\cite{kushilevitz:1998} {\em still holds} even if we assume multiple channels and/or collision detection---opening
an unexpected gap with the wake-up problem for which these assumptions improve the achievable time complexity.

\noindent {\bf Our Technique.}
Consider the following simple game which we call {\em $k$-hitting}.
A {\em referee} secretly selects a target set $T\subseteq \{1,2,...,k\}$.  
The game proceeds in rounds. In each round, a {\em player} (represented by a randomized algorithm) generates
a proposal $P$. If $|P\cap T|=1$, the player wins. Otherwise, it moves on to the next round.
In Section~\ref{sec:hitting}, we leverage a useful combinatorial result due to Alon et~al.~\cite{alon:1991}
to prove that this game requires $\Omega(\log^2{k})$ rounds to solve with high probability (w.r.t.~$k$),
and $\Omega(\log{k})$ rounds in expectation. (Notice, you could propose the sets of 
a {\em $(k,k)$-selective family}~\cite{clementi:2003} 
to solve this problem deterministically, but this would require $\Omega(k)$ proposals in the worst-case.)

These lower bounds are important because in this paper
we show that this basic hitting game reduces to solving
wake-up and broadcast under all of the different combinations of model assumptions that we consider.
In other words, whether or not you are solving wake-up or broadcast, assuming multiple channels or a single channel,
and/or assuming collision detection or no collision detection, if you can solve the problem fast you can solve this hitting game fast.
Our lower bounds on the hitting game, therefore, provide a fundamental speed-limit for basic communication tasks in the
radio network model.

The trick in applying this method is identifying the proper reduction argument for the assumptions in question.
Consider, for example,
 our reduction for wake-up with a single channel and no collision detection.
Assume some algorithm ${\cal A}$ solves wake-up with these assumptions in $f(n)$ rounds, in expectation.
As detailed in Section~\ref{sec:wakeup}, we can use ${\cal A}$ to define a player that solves
the $k$-hitting game in $f(k)$ rounds with the same probability---allowing the relevant hitting game lower bound to apply.
Our strategy for this case is to have the player simulate
${\cal A}$ running on all $k$ nodes in a network of size $k$. For each round of the simulation,
it proposes the ids of the nodes that broadcast, then simulates all nodes receiving nothing.
This is not necessarily a valid simulation of ${\cal A}$ running on $k$ nodes: {\em but it does not need to be.}
What we care about are the simulated nodes with ids in $T$: the (unknown to the player) target set for this instance of the hitting game.
The key observation is that in the {\em target execution} where only the nodes in $T$ are active, they will receive nothing until the first round where one node
broadcasts alone---solving wake-up. In the player's simulation, these same nodes are also receiving nothing (by the the player's fixed receive rule)
so they will behave the same way. This will lead to a round of the simulation where only one node from $T$ (and perhaps other nodes outside
of $T$) broadcast. The player will propose these ids,  winning the hitting game.

These reductions get more tricky as we add additional assumptions. Consider, for example, what happens
when we now assume collision detection. 
Maintaining consistency between the nodes in $T$ in the player simulation and the target execution
 becomes more complicated, as the player must now correctly simulate a collision event whenever two or more nodes
from $T$ broadcast---even though the player {\em does not know} $T$. Adding multiple channels only further complicates this
need for consistency. Each bound in this paper, therefore, is built around its own clever method for a hitting
game player to correctly
simulate a wake-up or broadcast algorithm in such a way that it wins the hitting game with the desired efficiency.
These arguments are simple to understand  and sometimes surprisingly elegant once identified, but can also be elusive
before they are first pinned down.

\noindent {\bf Roadmap.}
A full description of our results and how they compare to existing results is provided in Figure~\ref{fig:results}.
In addition, before we prove each bound in the sections that follow, we first 
discuss in more detail the relevant related work.
In Section~\ref{sec:model}, we formalize our model and the two problems we study.
In Section~\ref{sec:hitting}, we formalize and lower bound the hitting games at the core of our technique.
In Section~\ref{sec:sim}, we detail a general simulation strategy that we adopt in most
of our wake-up bounds (by isolating this general strategy in its own section we reduce redundancy).
Sections~\ref{sec:wakeup} to~\ref{sec:wakeup:cd:multi} contain our wake-up lower bounds,
and Section~\ref{sec:bcast} contains our broadcast lower bound.
(We only need one section for broadcast as we prove that the same result holds for all assumptions considered in this paper.)

%%%%%%%%%%%%%%%%%%%%%%%%%
\begin{figure}[t]
\centering
\begin{tabular}{|c||c|c|}
\hline
       & \bf Existing Results (exp. $\mid$ high)   & \bf This Paper (exp. $\mid$ high) \\ \hline \hline
\bf wake-up  &   $\Omega(\log{n})$ $\mid$ $\Omega(\log^2{n})$ \cite{jurdzinski:2002,colton:2006} &  $\Omega(\log{n})$ $\mid$ $\Omega(\log^2{n})$ (*) \\ \hline
\bf wake-up/cd  & $\Omega(\log\log{n})$ $\mid$ $\Omega(\log{n})$ \cite{willard:1986} & $\Omega(\log\log{n})$ $\mid$ $\Omega(\log{n})$ (*)  \\ \hline
\bf wake-up/mc  &  {\em (open)} $\mid$ $\Omega(\frac{\log^2{n}}{\chan\log\log{n}} + \log{n})$ \cite{dolev:2009,daum:2012b} & 
              $\Omega(\frac{\log{n}}{\chan} + 1)$ $\mid$ $\Omega(\frac{\log^2{n}}{\chan} + \log{n})$ (*)  \\ \hline
\bf wake-up/cd/mc  &  $\Omega(1)$ $\mid$ {\em (open)} & $\Omega(1)$ $\mid$ $\Omega(\frac{\log{n}}{\log{\chan}} + \log\log{n})$  \\ \hline \hline
\bf broadcast  &  $\Omega(D\log{(n/D)})$ \cite{kushilevitz:1998} & $\Omega(D\log{(n/D)})$  \\ \hline
\bf broadcast/cd/mc  & {\em (open)}  &  $\Omega(D\log{(n/D)}$ \\ \hline
\end{tabular}
\caption{\footnotesize This table summarizes the expected time (exp.) and high probability (high) results
for wake-up and broadcast in the existing literature as well as the new bounds proved in this paper.
In these bounds, $n$ is the network size, $\chan$ the number of channels, and $D$ the network diameter.
In the problem descriptions, ``cd" indicates the collision detection assumption and ``mc" indicates the multiple channels assumption.
In the existing results we provide citation for the paper(s) from which the results derive and use {\em ``(open)"} to indicate a previously open problem.
In all cases, the new results in this paper simplify the existing results. We marked some of our results with ``(*)" to indicate
that the existing results assumed the restricted {\em uniform} class of algorithms. All our algorithms hold for all randomized algorithms,
so any result marked by ``(*)" is strictly stronger than the existing result. We do not separate expected time and high probability
for the broadcast problems as the tight bounds are the same for both cases.}
\label{fig:results}
\end{figure}

%%%%%%%%%%%%%%%%%%

\section{Model and Problems}
\label{sec:model}

In this paper we consider variants
 of the standard {\em radio network model}.
This model represents a radio network with a connected undirected graph $G=(V,E)$ of diameter $D$.
The $n=|V|$ nodes in the graph represent the wireless devices and the edges in $E$ capture communication proximity.
In more detail, executions in this model proceed in synchronous rounds. In each round, each node can choose to either
{\em transmit} a message or {\em receive}. In a given round, a node $u\in V$ can receive a message
from a node $v\in V$, if and only if the following conditions hold: (1) $u$ is receiving and $v$ is transmitting;  (2)
$v$ is $u$'s neighbor in $G$; and (3) $v$ is the {\em only} neighbor of $u$ transmitting in this round.
The first condition captures the half-duplex nature of the radio channel and the second condition
captures message collisions. 
To achieve the strongest possible lower bounds, we assume nodes are provided unique ids from $[n]$.
In the following, we say an algorithm is {\em uniform} if (active) nodes use a predetermined sequence of independent
broadcast probabilities to determine whether or not to broadcast in each round, up until they first receive a message.

In the {\em collision detection} variant of the radio network model, a receiving node $u$ can distinguish between silence (no neighbor
is transmitting) and collision (two or more neighbors are transmitting) in a given round. 
In this paper, to achieve the strongest possible lower bounds, when studying
single hop networks we also assume that a {transmitter} can distinguish between broadcasting
alone and broadcasting simultaneously with one or more other nodes.
In the {\em multichannel} variant of the radio network model,
we use a parameter $\chan\geq 1$ to indicate the number of orthogonal communication
channels available to the nodes. When $\chan>1$, we generalize the model to require each node to choose in each round
a single channel on which to participate. The communication rules above apply separately to each channel.
In other words, a node $u$ receives a message from $v$ on channel $c$ in a given round, if and only if in this round: (1) $u$ receives on $c$ and $v$ transmits on $c$;
 (2) $v$ is a neighbor of $u$; and (3) no other neighbor of $u$ transmits on $c$.

We study both {\em expected time} and {\em high probability} results, where we define the latter to mean probability at least $1-\frac{1}{n}$.
We define the notation $[i,j]$, for integers $i \leq j$, to denote the range $\{i,i+1,...,j\}$,
and define $[i]$, for integer $i>0$, to denote $[1,i]$.

\noindent {\bf Problems.}
The {\em wake-up} problem assumes a single hop network consisting of {\em inactive} nodes.
At the beginning of the execution, an arbitrary subset of these nodes are {\em activated} by an adversary.
Inactive nodes can only listen,
while  
active nodes execute an arbitrary randomized algorithm.
We assume  that active nodes have no advance knowledge of the identities of the other active nodes.
The problem is solved in the first round in which an active node broadcasts alone (therefore {\em waking up}
the listening inactive nodes).
When considering a model with collision detection, 
we still require that an active node broadcasts alone to solve the problem
(e.g., to avoid triviality, we assume that the inactive nodes need to receive a message to wake-up,
and that simply
detecting a collision is not sufficient\footnote{Without this restriction, the problem is trivially solved
by just having all active nodes broadcast in the first round.}).
When considering multichannel networks,
we assume the inactive nodes are all listening on the same known {\em default} channel (say, channel $1$).
To solve the problem, therefore, now requires that an active node broadcast alone on the default channel.

The {\em broadcast} problem assumes a connected multihop graph.
At the beginning of the execution,
a single {\em source} node $u$ is provided a message $m$.
The problem is solved once every node in the network has received $m$.
We assume nodes do not have any advance knowledge of the network topology.
As is standard,
we assume that nodes are inactive (can only listen) until they first receive $m$.
As in the wake-up problem, detecting a collision alone is not sufficient to activate an inactive node,
and in multichannel networks, we assume inactive nodes all listen on the default channel.

\section{The $k$-Hitting Game}
\label{sec:hitting}

The $k$-{\em hitting game}, defined for some integer $k>1$, 
assumes a {\em player} that faces off against an {\em referee}.
At the beginning of the game, the referee secretly selects a {\em target set} $T \subseteq \{1,...,k\}$.
The game then proceeds in rounds.
In each round, the player generates a {\em proposal} $P \subseteq \{1,...,k\}$.
If $|P \cap T| = 1$, then the player wins the game.
Otherwise, the player moves on to the next round learning no information other than the fact that
its proposal failed.
We formalize both entities as probabilistic automata and assume the player
does not know the referee definition and the referee does not know
the player's random bits. Finally, we define the {\em restricted} $k$-hitting game 
to be a variant of the game where the target set is always of size two.

\noindent {\bf A Useful Combinatorial Result.}
Before proving lower bounds for our hitting game 
we cite an existing combinatorial result that will aid our arguments.
To simplify the presentation of this result, we first define some useful notation.
Fix some integer $\ell >0$.
Consider two sets $A\subseteq \{1,2,...,\ell\}$ and $B\subseteq \{1,2,...\ell\}$.
We say that $A$ {\em hits} $B$ if $|A \cap B| = 1$.
Let an {\em $\ell$-family} be a family of non-empty subsets of $\{1,2,...,\ell\}$.
The {\em size} of an $\ell$-family $\mathscr{A}$, sometimes noted as $|\mathscr{A}|$,
 is the number of sets in $\mathscr{A}$.
Fix two $\ell$-families $\mathscr{A}$ and $\mathscr{B}$.
We say $\mathscr{A}$ {\em hits} $\mathscr{B}$, if for every $B\in \mathscr{B}$ there
exists an $A\in \mathscr{A}$ such that $A$ hits $B$.
Using this notation, we can now present the result:

\vfill\eject

\begin{lemma}[\cite{alon:1991,ghaffari:2013}]
There exists a constant $\beta > 0$,
such that for any integer $\ell >1$,
these two results hold:

\begin{enumerate}

  \item There exists
            an $\ell$-family $\mathscr{R}$, where $|\mathscr{R}| \in O(\ell^8)$,
            such that for every $\ell$-family $\mathscr{H}$ that hits $\mathscr{R}$,
    $|\mathscr{H}| \in \Omega(\log^2{\ell})$.
    
  \item There exists
            an $\ell$-family $\mathscr{S}$, where $|\mathscr{S}| \in O(\ell^8)$,
such that for every $H \subseteq \{1,2,...,\ell\}$, $H$ hits at
  most a $(\frac{1}{\beta\log{(\ell)}})$-fraction of the sets in $\mathscr{S}$.
  
  \end{enumerate}
  
\label{lem:alon}
\end{lemma}

\noindent The first result from this lemma was proved in a 1991 paper by Alon et~al.~\cite{alon:1991}.
It was established using the probabilistic method and was then used to 
prove a $\Omega(\log^2{n})$ lower bound
on {\em centralized} broadcast solutions in the radio network model.
The second result is a straightforward consequence of the analysis
used in~\cite{alon:1991},
recently isolated and proved by Ghaffari et~al.~\cite{ghaffari:2013}.

%%%%%%%%%%%%%%%%%%%%%%%%%%%%%%%%%%%%%%%%%%%%%%%%%%%%%%%%
%%%%%%%%%%%%%%%%%%%%%%%%%%%%%%%%%%%%%%%%%%%%%%%%%%%%%%%%
%%%%%%%%%%%%%%%%%%%%%%%%%%%%%%%%%%%%%%%%%%%%%%%%%%%%%%%%

%
\noindent {\bf Lower Bounds for the $k$-Hitting Game.}
We now prove lower bounds on our general and restricted $k$-hitting games.
These results, which concern probabilities,
 leverage Lemma~\ref{lem:alon}, which concerns combinatorics,
  in an interesting way which depends on the size of $\R$ and $\mathscr{S}$ being polynomial in $\ell$.
 %The third result uses a counting argument to extended Lemma~\ref{lem:alon} in a manner
 %relevant to the restricted hitting
 %game, then applies an argument used to establish the high probability bound for non-restricted hitting. 
 
\begin{theorem}
Fix some player ${\cal P}$ that guarantees, for all $k > 1$,
to solve the $k$-hitting game in $f(k)$ rounds, in expectation.
It follows that $f(k) \in \Omega(\log{k})$.
\label{thm:hitting:expectation}
\end{theorem}
\begin{proof}
%Fix some $n$ sufficiently large to apply Theorem~\ref{thm:main}.
%Let $\R$ be the resulting separator.
%A consequence of property $2$ of Theorem~\ref{thm:main}
%is that there exists a constant $\alpha > 0$,
%such that every proposal $P\subseteq [n]$ hits at most a $\frac{1}{\alpha\log{n}}$-fraction of the sets in $\R$.
%It follows that any sequence of proposasl of length less than $\frac{\alpha\log{n}}{2}$ hits less than half of the sets in $\R$.
Fix any $k>1$.
Let $\beta$ and $\mathscr{S}$ be the constant and $\ell$-family
provided by the second result of Lemma~\ref{lem:alon} applied to $\ell=k$.
The lemma tells us that for any $P \subseteq[k]$, 
$P$ hits at most a $(\frac{1}{\beta\log{k}})$-fraction of the sets in $\mathscr{S}$.
It follows that for any $k$-family $\mathscr{H}$,
such that $|\mathscr{H}| < \frac{\beta\log{k}}{2}$,
$\mathscr{H}$ hits less than half the sets in $\mathscr{S}$.

We now use these observations to prove our theorem.
Let ${\cal P}$ be a $k$-hitting game player. 
Consider a referee that selects the target set by choosing a set $T$ from $\mathscr{S}$ with uniform randomness.
Let $\mathscr{H}$ be 
the first $\lfloor \frac{\beta\log{k}}{2} \rfloor - 1$ proposals generated by ${\cal P}$ in a given instance of the game.
By our above observation,
this sequence of proposals
hits less than half the sets in $\mathscr{S}$.
Because the target set was chosen from $\mathscr{S}$ with randomness that was uniform and independent
of the randomness used by ${\cal P}$ to generate its proposals, 
it follows that the probability that $\mathscr{H}$ hits the target is less than $1/2$.
To conclude, we note that $f(k)$ must therefore be larger than $\lfloor \frac{\beta\log{k}}{2} \rfloor - 1 \in \Omega(\log{k})$,
as required by the theorem.
\end{proof}

%\noindent We now extend this strategy to achieve a high probability bound on our game.

\begin{theorem}
Fix some player ${\cal P}$ that guarantees, for all $k > 1$,
to solve the $k$-hitting game in $f(k)$ rounds with probability at least $1-\frac{1}{k}$.
It follows that $f(k) \in \Omega(\log^2{k})$.
\label{thm:hitting:high}
\end{theorem}
\begin{proof}
Fix any $\ell>1$.
Let  $\mathscr{R}$ be the $\ell$-family
provided by the first result of Lemma~\ref{lem:alon} applied to this value.
Let $t=|\mathscr{R}|$. We know from the lemma that $t \in O(\ell^8)$.

%we fix some sufficiently large $n$ to apply Theorem~\ref{thm:main} and identify a separator $\R$
%that satisfies that theorem's properties. 
%Let $t$ by the size of $\R$.
%By Theorem~\ref{thm:main}, we know $t=c n^7 \log{n}$, for some positive constant $c$.

To achieve our bound, we will consider the behavior of a player ${\cal P}$
in the $k$-hitting game for $k = t + 1$.
As in Theorem~\ref{thm:hitting:expectation}, 
we have our referee select its target set by choosing a set from $\mathscr{R}$ with uniform randomness.
(Notice, in this case, our referee is actually making things {\em easier} for the player
by restricting its choices to only the values in $[\ell]$ even though the game is defined
for the value set $[k]$, which is larger. As we will show, this advantage does
not help the player much.)

Let $c\log^2{(\ell)}$, for some constant $c>0$,
be the exact lower bound from the first result of Lemma~\ref{lem:alon}.
Let
$\mathscr{H}$ be the first $\lfloor c\log^2{(\ell)} \rfloor - 1$ proposals
generated by ${\cal P}$ in a given instance of the game.
Lemma~\ref{lem:alon} tells us that there is at least one set $R\in \mathscr{R}$ that $\mathscr{H}$ does not hit.
Because the target set was chosen from $\mathscr{R}$ with randomness that was uniform and independent
of the randomness used by ${\cal P}$, 
it follows that the probability that $\mathscr{H}$ misses the target is at least $1/t$ (recall that $t$ is the size of $\mathscr{R}$).
Inverting this probability,
it follows that the probability that ${\cal P}$ wins the game with the proposals represented by $\mathscr{H}$ is less
than or equal to $1-\frac{1}{t} =1 - \frac{1}{k-1} < 1 - \frac{1}{k}$.
It follows that $f(k)$ must be larger than $|\mathscr{H}|$ and therefore
must be of size at least $c\log^2{(\ell)}\in \Omega(\log^2{(\ell)})$.
To conclude the proof, 
we note that $k \in O(\ell^8)$,
and therefore we can express $\ell$ in terms of $k$ as some polynomial in $\Theta(k^{1/d})$, for some positive constant $d \leq 8$.
Substituting for $\ell$ in our above equation,
it follows that $f(k)\in \Omega(\log^2{(\ell)}) \in \Omega(\log^2{(k^{1/d})}) \in \Omega(\log^2{(k)})$,
as required by the theorem.
\end{proof}

%\noindent The below theorem is proved similar to Theorem~\ref{thm:hitting:high}.
%The details are in the appendix.

\begin{theorem}
Fix some player ${\cal P}$ that guarantees, for all $k> 1$,
to solve the {\em restricted} $k$-hitting game in $f(k)$ rounds with probability at least $1-\frac{1}{k}$.
It follows that $f(k) \in \Omega(\log{k})$.
\label{thm:hitting:restricted}
\end{theorem}
\begin{proof}
Our proof strategy is to prove a variant of the first result of Lemma~\ref{lem:alon} that will allow
us to reuse the proof argument of Theorem~\ref{thm:hitting:high} to prove our needed result for
the restricted case.
To do so, 
fix any $k>1$.
Consider the $k$-family $\R_2$
that consists of the ${k \choose 2}$ unique pairs in $[k] \times [k]$.
Fix some $k$-family $\H$ of size $t < \log{k}$.
We now show the existence of some $R\in \R_2$ not hit by $\H$.
To do so, we first define a function $f_{\H}: [k] \rightarrow \{0,1\}^{t}$,
where $f_{\H}(i)$ returns a binary string $b_i$ of length $t$,
where bit $r$ of $b_i$ is $1$ if and only if $i$ is in the $r^{th}$ set in $\H$,
by some fixed ordering of these sets.
Given our assumption that $t < \log{k}$, 
it follows that the total number of unique binary strings of length $t$
can be upper bounded as 
$2^{t} < 2^{\log{k}} = k$.

The pigeonhole principle tells us that there exist $i,j\in [k], i \neq j$,
such that $f_{\H}(i) = f_{\H}(j)$. 
It follows that the set $\{i,j\} \in \R_2$ is not hit by $\H$,
as we just established that there is no $H\in \H$ that contains $i$ or $j$,
but not both.
At this point, we have established the existence of a $k$-family made up
only of sets of size two ($\R_2$) such that any $k$-family $\H$ that
hits this family must be of size at least $\log{k}$.
We can consider this a variant of the first result of Lemma~\ref{lem:alon},
and therefore achieve our bound by now applying
the same proof argument as in Theorem~\ref{thm:hitting:high} to this variant of the result.
This argument provides that $\Omega(\log{k})$ rounds are required to solve the restricted
hitting game with probability at least $1-\frac{1}{k}$, as required by the theorem.
%
%(Notice, general forms of this argument are common in combinatorics.
%An alternative proof of this property, for example,
%would be to leverage the $\Omega((k/24)\log{(n/k)})$ lower bound for the size of $(n,k)$-{\em selective
%families},
% established in~\cite{clementi:2003},
% applied to the case where $k=2$.)
%
%
\end{proof}

\section{Simulation Strategy}
\label{sec:sim}

Most of our bounds for the {\em wake-up} problem use a similar
simulation strategy. To reduce redundancy, we define
the basics of the strategy and its accompanying notation in its own section.
In more detail, the {\em wake-up simulation strategy}, defined with respect
to a wake-up algorithm ${\cal A}$, is a general strategy for
 a $k$-hitting game player to generate proposals based on a local simulation of ${\cal A}$.
The strategy works as follows. The player simulates ${\cal A}$ running on all $k$ nodes
in a $k$-node network satisfying the same assumptions on collision detection and channels assumed by ${\cal A}$.
For each simulated round, the player will generate one or more proposals for the hitting game.
In more detail, at the beginning of a new simulated round, the player simulates the $k$ nodes
running ${\cal A}$ up until the point that they make a broadcast decision.
At this point, the player applies a {\em proposal rule} that transforms these decisions into one or more
proposals for the hitting game.
The player then makes these proposals, one by one, in the game.
If none of these proposals wins the hitting game, then the player most complete the current simulated
round by using a {\em receive rule} to specify what each node receives;
 i.e., silence, a message, or a collision (if collision detection is assumed by ${\cal A}$).
In other words, a given application of the wake-up simulation strategy is defined by two things:
a definition of the {\em proposal rule} and {\em receive rule} used by the player to generate proposals from the simulation,
and specify receive behavior in the simulation, respectively.

To analyze a wake-up simulation strategy for a given instance of the $k$-hitting game with target set $T$,
we define the {\em target execution} for this execution to be the execution that would
result if ${\cal A}$ was run in a network where only the nodes corresponding to $T$ were active and they used
the same random bits as the player uses on their behalf in the simulation.
We say an instance of the simulation strategy is {\em consistent} with its target execution through a given round,
if the nodes corresponding to $T$ in the simulation behave the same (e.g., send and receive the same messages)
as the corresponding nodes in the target execution through this round.

\section{Lower Bounds for Wake-Up}
\label{sec:wakeup} 

We begin by proving tight lower bounds for both expected and high probability solutions
to the wake-up problem in the most standard
set of assumptions used with the radio network model: a single channel and no collision detection.
As explained below, our bounds are tight and generalize the best know previous bounds,
which hold only for uniform algorithms, to now apply to all randomized algorithms.
(We note that a preliminary version of our high probability bound below appeared
as an aside in our previous work on structuring multichannel radio networks~\cite{daum:2012}).

In terms of related work, the {\em decay} strategy introduced Bar-Yehuda et~al.~\cite{baryehuda:1992} solves the
wake-up problem in this setting with high probability in $O(\log^2{n})$ rounds and in expectation in $O(\log{n})$ rounds.
In 2002, Jurdzinski and Stachowiak~\cite{jurdzinski:2002}
proved the necessity of $\Omega\big(  \frac{\log{n}\log{(1/\epsilon)}}{\log\log{n} + \log\log{(1/\epsilon)}} \big)$
rounds to solve wake-up with probability at least $1-\epsilon$, which proves decay optimal
within a $\log\log{n}$ factor.
Four years later, Farach-Colton et al.~\cite{colton:2006} removed the $\log\log{n}$ factor
by applying linear programming techniques.
As mentioned, these existing bounds only apply to {uniform} algorithms in which nodes 
use a predetermined sequence of broadcast probabilities. (Section $3.1$ of~\cite{colton:2006}
claims to extend their result to a slightly more general class of uniform algorithms
in which a node can choose a uniform algorithm to run based on its unique id.)
%\footnote{The main
 %idea of the claim (made in Section $3.1$ of~\cite{colton:2006}),
%is that given an algorithm ${\cal A}$ that uses
%its unique id to choose a broadcast schedule, you could simulate ${\cal A}$
%with a uniform algorithm ${\cal A'}$ that randomly generates a unique id then applies
%the logic of ${\cal A}$. It is not made clear, however, why ${\cal A'}$ should be considered
%uniform. This algorithm uses bits generated early in an execution to determine broadcast
%probabilities used later in the execution, creating a dependency between broadcast
%events that violates the definition of uniformity as we understand it.
%On the other hand, a closer inspection of the proof technique of this paper
%does not reveal any fundamental reasons why it could not also be applied to this
%slightly more general class of uniform algorithms, so the claim is likely true.})

%\paragraph{High Probability.}
%The key idea in this proof is the same key idea we will use in all lower bounds that follow:
%if you can solve wake-up fast, then a player can solve hitting game fast by simulating
%the wake-up algorithm to determine its guess.

\begin{theorem}
Let ${\cal A}$ be an algorithm that solves wake-up with high probability in $f(n)$ rounds
in the radio network model with a single channel and no collision detection.  It follows
that $f(n) \in \Omega(\log^2{n})$.
\label{thm:wakeup:high}
\end{theorem}
\begin{proof}
Fix some wake-up algorithm ${\cal A}$ that solves wake-up in $f(n)$ rounds
with high probability in a network with one channel and no collision detection.
We start by defining a wake-up simulation strategy that uses ${\cal A}$ (see Section~\ref{sec:sim}).
In particular,
consider the {\em proposal rule} that has the player propose the id of every node that broadcasts
in the current simulated round,
and the {\em receive rule} that always has all nodes receive nothing.

Let ${\cal P}_{\cal A}$ be the $k$-hitting game player that uses this simulation strategy.
We argue that ${\cal P}_{\cal A}$ solves the $k$-hitting game in $f(k)$ rounds with high probability in $k$.
To see why, notice that for a given instance of the hitting game with target $T$,
${\cal P}_{\cal A}$ is consistent with the target execution until the receive rule of the first round
in which exactly one node in $T$ broadcasts.
(In all previous rounds, ${\cal P}_{\cal A}$ correctly simulates the nodes in $T$ receiving
nothing, as its receive rule has all nodes always receive nothing.)
Assume ${\cal A}$ solves wake-up in round $r$ in the target execution.
It follows that $r$ is the first round in which a node in $T$ broadcasts alone in this execution.
By our above assumption, ${\cal P}_{\cal A}$ is consistent with the target execution up to the application
of the receive rule in $r$. In particular, it is consistent when it applies the proposal rule
for simulated round $r$. By assumption, this proposal will include exactly one node from $T$---winning the hitting game.

We assumed that ${\cal A}$ solves wake-up in $f(n)$ rounds with high probability in $n$.
Combined with our above argument, it follows
that ${\cal P}_{\cal A}$ solves the $k$-hitting game in $f(k)$ rounds with high probability in $k$.
To complete our lower bound, 
we apply a contradiction argument.
In particular, assume for contradiction that there exists a wake-up algorithm ${\cal A}$
that solves wake-up in $f(n) \in o(\log^2{n})$ rounds, with high probability.
The hitting game player ${\cal P}_{\cal A}$ defined above
will therefore solve $k$-hitting in $o(\log^2{n})$ rounds with high probability.
This contradicts Theorem~\ref{thm:hitting:high}.
\end{proof}

\begin{theorem}
Let ${\cal A}$ be an algorithm that solves wake-up in $f(n)$ rounds, in expectation,
in the radio network model with a single channel and no collision detection.  It follows
that $f(n) \in \Omega(\log{n})$.
\label{thm:wakeup:expectation}
\end{theorem}
\begin{proof}[Proof Idea.]
It is sufficient to apply the same argument as in Theorem~\ref{thm:wakeup:high}.
The only change is in the final contradiction argument, where we simply replace $\log^2{n}$ with $\log{n}$,
and now contradict Theorem~\ref{thm:hitting:expectation}.
\end{proof}

%%%%%%%%%%%%%%%%%%%%%%%%%%%%%%%%%%%%%%%%%%%%%%%%%%%%%%%%
%%%%%%%%%%%%%%%%%%%%%%%%%%%%%%%%%%%%%%%%%%%%%%%%%%%%%%%%
%%%%%%%%%%%%%%%%%%%%%%%%%%%%%%%%%%%%%%%%%%%%%%%%%%%%%%%%

\section{Lower Bounds for Wake-Up with Collision Detection}
\label{sec:wakeup:cd}

We prove tight lower bounds for expected and high probability bounds on the wake-up problem in the radio network model 
with collision detection.
%Also notice that these lower bounds apply to related problems, and in particular,
%{\em leader election}, as solving leader election solves wake-up (the leader, once elected, can simply broadcast
%alone to wake-up any nodes still inactive).
%We emphasize in particular our lower bound for the expected case
%which simplifies the seminal bound of Willard~\cite{willard:1986} and strengthens it to hold for all randomized
%algorithms and not just uniform solutions.
%
In terms of related work,
a seminal paper by Willard~\cite{willard:1986}
describes a wake-up algorithm (he called the problem
``selection resolution," but the definition in this setting is functionally identical) which solves
the problem in $O(\log\log{n})$ rounds, in expectation. He also proved the result tight
with an $\Omega(\log\log{n})$ lower bound for uniform algorithms.
As Willard himself admits, his lower bound proof is mathematically complex.
Below, we significantly simplify this bound and generalize it to hold for all algorithms. 
From a high-probability perspective, many solutions exist in folklore for solving wake-up (and related problems)
in $O(\log{n})$ rounds. Indeed, leveraging collision detection,
wake-up can be solved {\em deterministically} in $O(\log{n})$ rounds (e.g., use
the detector to allow the active nodes to move consistently through a binary search
tree to identify the smallest active id).
The necessity of $\Omega(\log{n})$ rounds seems also to exist in folklore.
%We struggled to identify a paper that proves this paper for
%the explicit purpose of establishing the fundamental limits of wake-up with collision detection
%(but we do not doubt that such a paper exists).

We begin with our high probability result.
Our simulation strategy is more difficult to deploy here because the player
must now somehow correctly simulate the collision detection among the nodes in
the (unknown) target set $T$. 
To overcome this difficulty,
we apply our solution to networks in which only two nodes are activated and then achieve a contradiction
with our lower bound on {\em restricted} hitting. 
%Due to space constraints, the details of this proof are deferred to the appendix.
%We begin by reproving the $\Omega(\log\log{n})$ round lower bound for expected time
%wake-up that follows directly from the seminal paper on leader election on multi-access channels
%by Willard~\cite{willard} (lower bounds for leader election apply to wake-up in this case
%as electing a leader is sufficient to solve wake-up: simply have the leader broadcast alone once elected).
%We then reprove that this bound increases significantly to $\Omega(\log{n})$ to achieve high probability.
%({\em Discuss what bounds these reprove and the complication of the existing.}

\begin{theorem}
Let ${\cal A}$ be an algorithm that solves wake-up with high probability in $f(n)$ rounds
in the radio network model with a single channel and collision detection.  It follows
that $f(n) \in \Omega(\log{n})$.
\label{thm:wakeup:cd:high}
\end{theorem}
\begin{proof}
Fix some wake-up algorithm ${\cal A}$ that solves wake-up in $f(n)$ rounds
with high probability in a network with one channel and collision detection.
We start by defining a wake-up simulation strategy that uses ${\cal A}$ (see Section~\ref{sec:sim}).
In particular,
consider the {\em proposal rule} that has the player propose the id of every node that broadcasts
in the current simulated round,
and a {\em receive rule} that has two cases:
(1) if a given player broadcast in the current simulated round, it is simulated as detecting a collision;
(2) if a given player did not broadcast in the current simulated round, it is simulated as receiving and detecting nothing.

Let ${\cal P}_{\cal A}$ be the restricted $k$-hitting game player that uses this simulation strategy.
We cannot argue that this player solves the general $k$-hitting game, 
as the receive rule above is not likely to be consistent for many target sets.
We instead argue
that ${\cal P}_{\cal A}$ solves {\em restricted} $k$-hitting in $f(k)$ rounds with high probability in $k$.
In other words, our receive rule above, we will show, keeps the simulation consistent when the target only contains two nodes 
(as is the case in restricted hitting).
In more detail, 
fix a given instance of the restricted $k$-hitting game with some target set $T=\{i,j\}$.
We argue that
${\cal P}_{\cal A}$ is consistent with the target execution until it applies the receive rule in the first round
in which a node in $T$ broadcasts alone (at which point, the player will have won the hitting game).
In particular, there are three cases relevant to the receive behavior in a given round of the target execution for $T$.
The first case is that $i$ and $j$ are both silent. In this case, they would both receive and detect
nothing in the target execution. By definition, they will both receive nothing in ${\cal P}_{\cal A}$'s simulation as well.
The second case is that both $i$ and $j$ broadcast. In this case, they would both correctly detect a collision in the target
execution. By definition, the same occurs in the simulation.
The third case has exactly one of the two nodes broadcasting.
In this case, the player wins the hitting game during the proposal rule of this simulated round, so we do not
have to care about applying the receive rule in a way that maintains consistency.

We assumed that ${\cal A}$ solves wake-up in $f(n)$ rounds with high probability in $n$.
Clearly, this bound still holds even if we restrict our attention to networks with only two nodes activated.
Combined with our above argument, therefore, it follows
that ${\cal P}_{\cal A}$ solves the restricted $k$-hitting game in $f(k)$ rounds with high probability in $k$.
Assume for contradiction that $f(n) \in o(\log{n})$.
It would follow that ${\cal P}_{\cal A}$ solves the restricted hitting game  in $o(\log{k})$ rounds with high probability.
This contradicts Theorem~\ref{thm:hitting:restricted}.
\end{proof}

We now simplify and strengthen Willard's bound of $\Omega(\log\log{n})$ rounds for expected time wake up. %
At the core of our result is a pleasingly simple but surprisingly useful observation:
if you can solve wake-up in $t$ rounds with collision detection, you can then use this strategy 
to solve the hitting game in $2^t$ rounds
by simulating (carefully) all possible sequences of outcomes for the collision detector behavior in a $t$ round execution.
Solving the problem in $o(\log\log{n})$ rounds (in expectation) with collision detection,
therefore,
yields a hitting game solution that requires
only $2^{o(\log\log{k})} = o(\log{k})$ rounds (in expectation), contradicting Theorem~\ref{thm:hitting:expectation}---our lower bound
on expected time solutions to the hitting game.

\begin{theorem}
Let ${\cal A}$ be an algorithm that solves wake-up in $f(n)$ rounds, in expectation,
in the radio network model with a single channel and collision detection.  It follows
that $f(n) \in \Omega(\log\log{n})$.
\label{thm:wakeup:cd:expectation}
\end{theorem}
 \begin{proof}
 Fix some algorithm ${\cal A}$ that solves wake-up in $f(n)$ rounds, in expectation, in this setting.
 We start by defining a player ${\cal P}_{\cal A}$ that simulates ${\cal A}$
 to solve $k$-hitting in no more than $2^{f(k)+1}$ rounds, in expectation.
 Our player will use a variant of the simulation strategy defined in Section~\ref{sec:sim} and used
 in the preceding proofs, and we will, therefore, adopt much of the terminology of this approach (with some minor modifications). 
 In more detail,
 in this variant, ${\cal P}_{\cal A}$ will 
 run a different fixed-length simulation of ${\cal A}$, starting from round $1$,
 to generate each of its guesses in the hitting game.
 Most of these simulations will {\em not} be consistent with the relevant target execution.
 We will show, however, that in the case that the target execution solves wake-up,
 at least one such simulation is consistent and will therefore win the game.

 In more detail,
 for a given $k$,
let $B_{f(k)}$ be a full rooted binary tree of depth $f(k)$.
 We define a tree node labeling $\ell$,
 such that for every non-root node $u$,
 $\ell(u) = 0$ if $u$ is a left child of its parent and $\ell(u) = 1$ if $u$ is a right child (by some consistent orientation).
 Let $d$ be the depth function (i.e., $d(u)$ is the depth of $u$ in the tree with $d(root)=0$).
 Finally, let $p(u)$ return the $d(u)$-bit binary string defined by the sequence of labels (by $\ell$)
 on the path that descends from the root to $u$ (including $u$).
 For example, if the path from the root to $u$ goes from the root to its right child $v$,
 then from $v$ to its left child $u$, $p(u) = 10$.

 Our player ${\cal P}_{\cal A}$, when playing the $k$-hitting game, generate one guess 
 for each node in $B_{f(k)}$.
 Fix some such node $u$.
 To generate a guess for $u$, the player first executes a $d(u)$-round
 simulation of ${\cal A}$, running on all $k$ nodes in a $k$-node network,
 using $p(u)$ to specify collision detector behavior (in a manner described below).
 After it simulates these $d(u)$ full rounds, 
 it then simulates just enough of round $d(u)+1$ to determine the simulated nodes' broadcast decisions in this round.
 The player proposes the id of the nodes that choose to broadcast in this final partial round.
 (When generating a guess for the root node, the player simply proposes the nodes that broadcast in the first round.)
 
 In more detail, for each round $r\leq d(u)$ of the simulation
 for tree node $u$, if the $r^{th}$ bit of $p(u)$ is $0$,
 the player simulates all nodes detecting silence,
 and if the bit is $1$, it simulates all nodes detecting a collision.
 As a final technicality,
 let $\kappa$ be the random bits provided to the player 
 to resolve its random choices.
 We assume that for each simulated node $i$,
 the players uses the same bits from $\kappa$ for $i$ in each of its simulations.
 We do not, therefore, assume independence between different simulations.
 
 Consider the target execution of ${\cal A}$ for a given instance of the hitting game with target set $T$ and random bits $\kappa$.
 Assume that the target execution defined for these bits and target set solves wake-up in some round $r \leq f(k)$.
 Notice that in every round $r' <r$,
 there are only two possible behaviors: (1) no nodes broadcast (and all nodes therefore receive and detect nothing); and (2) two or
 more nodes broadcast (and all nodes therefore detect a collision). 
By definition, there exists a node $u$ in $B_{f(k)}$ such that $p(u)$ is a binary string of length $r-1$,
 where for each bit position $i$
in the string, $i=0$ if no nodes broadcast in that round of the target execution,
and $i=1$ if two or more nodes broadcast in that round of the target execution.
It follows that the first $r-1$ rounds of the simulation associated with tree node $u$ are consistent with the target execution.
Because exactly one node from $T$ broadcasts in round $r$ of the target execution,
and the $u$-simulation is consistent through round $r-1$,
then this same single node from $T$ will broadcast in the simulated beginning of round $r$.
The player's proposal associated with $u$ will therefore win the hitting game.

Pulling together the pieces, by assumption, 
the target execution for a given $T$ and $\kappa$ solves wake-up in $f(k)$ rounds, in expectation.
It follows that our player solves $k$-hitting with the same probability.
The number of guesses required to solve the problem in this case is bounded by
the number of nodes in $B_{f(k)}$ (as there is one guess per node),
which is $2^{f(k)+1}-1$.
We can now conclude with our standard style of contradiction argument.
Assume for contradiction that there exists an algorithm ${\cal A}$
that solves wake-up with a single channel and collision detection in $f(n) \in o(\log\log{n})$ rounds, in expectation.
It follows that ${\cal P}_{\cal A}$ wins the $k$-hitting game in $2^{f(k)+1} \in o(\log{k})$ rounds,
%\footnote{To
%remind ourselves why this is true, consider, for example, that because $f(n)\in o(\log\log{n})$,
%for every positive $c$, there exists an $n_c$ where for $n>n_c$, $f(n) < c\cdot\log\log{n}$. 
%Now consider $c=1/2$.
%For this value of $c$, $2^{f(k)}\in o(\log{k})$.
%} 
in expectation.
This contradicts Theorem~\ref{thm:hitting:expectation}.
 \end{proof}

\section{Lower Bounds for Wake-Up with Multiple Channels}
\label{sec:wakeup:multi}

In recent years,
theoreticians have paid increasing attention to multichannel versions of the radio network model 
(e.g.,~\cite{dolev:2007,dolev:2008,gilbert:2009,dolev:2009,dolev:2011,daum:2012,daum:2012b}).
These investigations are motivated by the reality that most network cards allow the device to choose its channel
from among multiple available channels. From a theoretical perspective, the interesting question is how to
leverage the parallelism inherent in multiple channels to improve time complexity for basic communication problems.
Daum et~al.~\cite{daum:2012b}, building on results 
from Dolev et~al.~\cite{dolev:2009},
prove a lower bound
of $\Omega\big(\frac{\log^2{n}}{\chan\log\log{n}} + \log{n}\big)$ rounds
for solving wake-up with high probability and uniform algorithms in a network with $\chan$ channels.
A lower bound for expected-time solutions was left open.
The best known upper bound solves the problem in $O\big(\frac{\log^2{n}}{\chan} + \log{n}\big)$
rounds with high probability
 and in $O\big(\frac{\log{n}}{\chan} + 1\big)$ rounds in expectation~\cite{daum:2012b}.

In the theorems that follow, we prove new lower bounds that match the best known upper bounds.
These bounds close the $\log\log{n}$ gap that exists with the best known previous results,
 establish the first non-trivial expected time bound, and strengthen the results to hold for all algorithms.
 %
%Our new bounds further strengthen these existing results in that they hold for any randomized algorithm
%and not just uniform solutions. 
%Notice, in the asymptotic expressions that follow,
%we always first fix a specific value of $\chan$,
%and therefore allow the asymptotics to be defined with respect to an increasing $n$.
%
%To prove our high probability bound, both terms in the sum are tackled separately.
%To prove the first term, we show that a player can simulate an algorithm using $\chan$
%channels by making $\chan$ proposals for each simulated round---one for each channel---to
%test if $T$ has an isolated broadcast on any channel. The second term uses a reduction from the 
%restricted hitting game.
%The expected time result adopts a similar strategy as the first term.
%
%The proofs for these theorems are deferred to the appendix.
%
We begin with the high probability result. 
In this bound,
the $\log^2{n}/\chan$ term dominates when $\chan$ is small and the $\log{n}$ term dominates
when $\chan$ is large. We handle these cases separately in their own lemmas---each using a different
simulation strategy---then combine them to achieve our final theorem.

%To prove our main theorem,
 %we prove a separate lemma for each case.
%The $\log^2{n}/\chan$ lemma leverages the key idea that a hitting game player can
%simulate wake-up algorithm that uses $\chan$ channels by using $\chan$ proposals
%per simulated round (one for each channel).
%The $\log{n}$ term leverages the key insight that when we restrict
%our attention to only $2$ nodes, we a hitting game player can successfully simulate a multi-channel
%wake-up algorithm with only a constant-time overhead.

\begin{lemma}
Let ${\cal A}$ be an algorithm that solves wake-up with high probability in $f(n,\chan)$ rounds
in the radio network model with $\chan \geq 1$ channels.  It follows
that for every $\chan\geq 1$, $f(n,\chan) \in \Omega(\log^2{n}/\chan)$.
\label{lem:multi:1}
\end{lemma}
\begin{proof}
Assume for contradiction that there exists a
wake-up algorithm ${\cal A}$ that solves wake-up with high probability
in this setting in $o(\log^2{n}/\chan)$ rounds for some $\chan$.
% (notice, this is only possible when $\chan \in O(\log{n})$). 
We start by defining a wake-up simulation strategy (see Section~\ref{sec:sim})
that simulates ${\cal A}$ in a network with $\chan$ channels. 
In particular, consider the {\em proposal rule} that generate $\chan$ proposals
for each simulated round. In particular, for each channel $c\in [\chan]$,
it proposes the ids of the simulated nodes (if any) that broadcast on $c$ in the simulated round.
%It follows that up to $\chan$ proposals are generated for each simulated round.
Assume the player uses the simple {\em receive rule} that has all simulated nodes receive nothing.

Let ${\cal P}_{\cal A, C}$ be the $k$-hitting game player that deploys this simulation strategy for our fixed value of $\chan$.
We argue that ${\cal P}_{\cal A, C}$ solves $k$-hitting in $f(k,\chan)\cdot\chan$ rounds with high probability.
To do so, we first argue that for a given instance of the hitting game with some target set $T$, 
the player is consistent with the target execution until it applies the receive rule in
the first round in which there is a channel on which some node from $T$ broadcasts alone,
as in all previous rounds and channels, there are either no broadcasters or two or more broadcasters from $T$: both
cases in which the receive rule behavior of receiving nothing is correct.
(This event does not necessarily imply that the player's simulation becomes {\em inconsistent}---for example,
if a node from $T$ broadcasts by itself on some channel other than $1$ with no other nodes from $T$ present to receive, the simulation is still consistent---but
it holds that before this event happens the simulation is definitely consistent.)

Next, assume ${\cal A}$ solves wake-up in round $r$ of this target execution.
This requires a node from $T$ to broadcast alone on channel $1$ in $r$.
It follows that in some round $r' \leq r$ in this target execution some
node from $T$ broadcasts alone on a channel for the first time.
As argued above, our simulation strategy is consistent with the target execution through $r'-1$.
Therefore, the simulation will make the same broadcast decisions for nodes in round $r'$
as in the target execution.
Let $c$ be the smallest channel with a single broadcaster from $T$ in $r'$ of the target execution.
When the player makes its proposal for this channel and this round,
it will win the hitting game.

We assumed that ${\cal A}$ solves wake-up in $f(n,\chan)$ rounds with high probability in $n$.
Combined with our above argument, it follows that our player ${\cal P}_{\cal A, C}$
solves $k$-hitting in $\leq f(k,\chan)\cdot\chan$ rounds with high probability in $k$. 
By our assumption, however, $f(k,\chan) \in o(\log^2{k}/\chan)$ for our fixed value of $\chan$,
which implies $f(k,\chan)\cdot \chan \in o(\log^2{k})$. This contradicts Theorem~\ref{thm:hitting:high}.
%
%
%To complete our lower bound, we apply a contradiction argument. In particular, assume for contradiction
%that there exists an algorithm ${\cal A}$ and channel number $\chan$,
%such that ${\cal A}$ solves wake-up in $f(k,\chan)\in o(\log^2{n}/\chan)$ rounds with high probability.
%The hitting game player ${\cal P}_{\cal A, C}$ will therefore solve the $k$-hitting
%game in $f(k,\chan)\cdot\chan = o(\log^2{k})$ rounds with high probability.
%This contradicts Theorem~\ref{thm:hitting:high}.
\end{proof}

%\noindent We now tackle the $\log{n}$ term that dominates to the case where $\chan$ is large.
%The key insight in the following argument is to once again focus on reducing from the restricted hitting game,
%an easier game that is also easier to simulate consistently. 

\begin{lemma}
Let ${\cal A}$ be an algorithm that solves wake-up with high probability in $f(n,\chan)$ rounds
in the radio network model with $\chan \geq 1$ channels.  It follows
that for every $\chan \geq 1$, $f(n,\chan) \in \Omega(\log{n})$.
\label{lem:multi:2}
\end{lemma}
\begin{proof}
Assume for contradiction
that there exists a wake-up algorithm ${\cal A}$ that solves wake-up with high probability
in this setting in $o(\log{n})$ rounds for some $\chan$. 
We start by defining a wake-up simulation strategy (see Section~\ref{sec:sim})
that simulates ${\cal A}$ in network with $\chan$ channels.
In particular,
consider the {\em proposal rule} that generates up to two proposals
per simulated round:  the first proposal includes the ids of every node (if any) that broadcast in this simulated round (regardless
of their channel choice), while the second proposal includes only the ids of every node (if any) that broadcast {\em on channel $1$} during
this simulated round. 
Assume the player uses the simple {\em receive rule} that has all nodes always receive nothing.

Let ${\cal P}_{\cal A, C}$ be the $k$-hitting game player that uses this simulation strategy.
We cannot prove that this player generates a simulation consistent with the target execution
for all possible target sets.
For our purposes here, however,
we only need prove that the game is consistent in the special case where the size of the target set is always of size two
(clearly, an algorithm that works for all network sizes will work in the special case where the number of active nodes happens to be two).
We will then derive our contradiction with the lower bound on the {\em restricted} hitting game, which is sufficiently strong to achieve
our needed logarithmic result.

In more detail,
fix an an instance of the restricted hitting game with some target set $T=\{i,j\}$.
We call a round of the target execution {\em meaningful} if at least one of the two following
conditions holds: (1) exactly one node from $T$ broadcasts on channel $1$; (2) exactly one node from $T$ broadcasts.
Notice, these are not equivalent conditions. If, for example, $i$ broadcasts on channel $1$ and $j$ on
channel $2$, we satisfy the first property but not the second.
We first argue that the player is consistent with the target execution
until the receive rule is applied in the first meaningful round.
To do so, consider the different combinations of possible behavior for $i$ and $j$ in a non-meaningful round:
if $i$ and $j$ are both silent in a given round of the target execution, they both receive nothing
in the target execution and in the player's simulation;
if $i$ and $j$ both broadcast in the target execution, and it is not the case that exactly one of these two nodes broadcasts on channel $1$,
then both receive nothing and wake-up is not solved in the target execution as well as in the simulation.

Assume the target execution eventually generates a meaningful round. Call this round $r$.
The player wins the hitting game in round $r$.
This follows because we argued that the player is consistent with the target execution through $r-1$.
Therefore, it will simulate the same broadcast behavior in $r$ as in the target execution.
Regardless of which case in the definition of meaningful applies in $r$,
one of the proposals for this round will win the game.
Pulling together the pieces, we note that the algorithm ${\cal A}$ solves wake-up (and therefore generates a meaningful round) in this setting in $f(n,\chan)$ rounds, with high probability in $n$, for our fixed $\chan$.
Therefore, our player solves restricted $k$-hitting in no more than $2f(k,\chan)$ rounds, with high probability in $k$.
We assumed, however, that $f(n,\chan) \in o(\log{n})$.
It follows that our player solves the restricted hitting game in $\leq 2(k,\chan) \in o(\log{k})$ rounds with high probability in $k$.
This contradicts Theorem~\ref{thm:hitting:restricted}.
\end{proof}

 Our main theorem follows directly from Lemma~\ref{lem:multi:1} and~\ref{lem:multi:2}:

\begin{theorem}
Let ${\cal A}$ be an algorithm that solves wake-up with high probability in $f(n,\chan)$ rounds
in the radio network model with $\chan \geq 1$ channels.  It follows
that for every $\chan \geq 1$, $f(n,\chan) \in \Omega(\log^2{n}/\chan + \log{n})$.
\label{thm:wakeup:multi:high}
\end{theorem}

 Moving on to the expected case,
we prove the necessity of $\log{n}/\chan$ rounds
using the same technique as in Lemma~\ref{lem:multi:1}:

\begin{theorem}
Let ${\cal A}$ be an algorithm that solves wake-up in $f(n,\chan)$ rounds, in expectation,
in the radio network model with $\chan \geq 1$ channels.  It follows
that for every $\chan \geq 1$, $f(n,\chan) \in \Omega(\log{n}/\chan + 1)$.
\label{thm:wakeup:multi:expectation}
\end{theorem}
\begin{proof}[Proof Idea.]
We can apply the same
wake-up simulation strategy and
analysis as in Lemma~\ref{lem:multi:1}.
In this case, we are simply replacing $\log^2{n}$ with $\log{n}$,
and now deriving our contradiction with  Theorem~\ref{thm:hitting:expectation}.
 \end{proof}

 \section{Lower Bound for Wake-Up With Collision Detection and Multiple\\ Channels}
 \label{sec:wakeup:cd:multi}
 
The final combination of model parameters to consider for wake-up is collision detection
{\em and} multiple channels.
No non-trivial upper or  lower bounds 
are currently known for this case.
We rectify this omission by proving below that $\Omega(\log{n}/\log{\chan} + \log\log{n})$
rounds are necessary to solve this problem with high probability in this setting.
Notice, this bound represents an interesting split with the preceding multichannel
results (which assume no collision detection), as the speed-up is now logarithmic in $\chan$ instead of linear.
On the other hand, the $\log^2{n}$ term in the previous case is replaced here with a faster $\log{n}$ term.

Collision detection, in other words, seems to be powerful enough on its own that adding extra channels
does not yield much extra complexity gains.
We do not consider an expected time result for this setting.
This is because even {\em without} collision detection,
the best known upper bound for multichannel networks~\cite{daum:2012b}
 approaches $O(1)$ time (which is trivially optimal)
quickly as the number of channels increases.
 %The proof for the below theorem, which combines techniques
 %from both Section~\ref{sec:wakeup:cd} and Section~\ref{sec:wakeup:multi},
 %is deferred to the appendix.

\begin{theorem}
Let ${\cal A}$ be an algorithm that solves wake-up with high probability in $f(n,\chan)$ rounds
in the radio network model with $\chan \geq 1$ channels and collision detection.  It follows
that for every $\chan \geq 1$, $f(n,\chan) \in \Omega(\log{n}/\log{\chan} + \log\log{n})$.
\label{thm:wakeup:multicd:high}
\end{theorem}
 %%%%%%%%%%%%%%%%%%%%%%%%%%%%%%%%%%%%%%%%%%%%%%%%%%%%%%%%
%%%%%%%%%%%%%%%%%%%%%%%%%%%%%%%%%%%%%%%%%%%%%%%%%%%%%%%%
%%%%%%%%%%%%%%%%%%%%%%%%%%%%%%%%%%%%%%%%%%%%%%%%%%%%%%%%
% Moved to Appendix
%%%%%%%%%%%%%%%%%%%%%%%%%%%%%%%%%%%%%%%%%%%%%%%%%%%%%%%%
%%%%%%%%%%%%%%%%%%%%%%%%%%%%%%%%%%%%%%%%%%%%%%%%%%%%%%%%
%%%%%%%%%%%%%%%%%%%%%%%%%%%%%%%%%%%%%%%%%%%%%%%%%%%%%%%%
\begin{proof}
Assume for contradiction
that there exists a wake-up algorithm ${\cal A}$ for this setting
that solves the problem $o(\log{n}/\log{\chan} + \log\log{n})$ with high 
probability in $n$, for some $\chan\geq 1$.
%
%and channel count $\chan \geq 1$.
%Assume ${\cal A}$ solves wakeup in $f(n,\chan)$ rounds with high probability.
%To prove $f(n,\chan) \in \Omega(\log{n}/\log{\chan} + \log\log{n})$,
%we handle the first and second term individually.

To achieve our final bound, we will handle both the $\log{n}/\log{\chan}$
and the $\log\log{n}$ term separately.
We begin with the $\log{n}/\log{\chan}$ term.
This term is non-trivial only when $\chan < n$,
so assume this holds for the following argument.
Consider the {\em restricted} wake-up problem where
the adversary guarantees to activate exactly two nodes.
We
can construct a new wake-up algorithm, ${\cal A'}$,
that simulates ${\cal A}$ running in a multichannel
network to solve restricted wake-up in a network with collision detection and 
only a {\em single} channel.
We will then use ${\cal A'}$ to solve the hitting game in a manner that generates a contradiction.

In more detail, the two nodes, $i'$ and $j'$,
running ${\cal A'}$ in an instance of the restricted wake-up problem,
will work together to simulate two nodes, $i$ and $j$, 
running ${\cal A}$ in a network with $\chan$ channels.
To implement this simulation, $i'$ keeps the simulated
state of $i$ and $j'$ keeps the state of $j$.
To maintain consistency, $i'$ and $j'$ use a {\em group}
of $\lceil \log{\chan}\rceil  + 1$ rounds to simulate each round of $i$ and $j$ running ${\cal A}$.

At the beginning of each such group, 
$i'$ and $j'$ advance their simulation of ${\cal A}$ just
far enough to determine the channel choice and broadcast behavior of $i$ and $j$, respectively.
At this point, $i'$ and $j'$ must coordinate their simulation to ensure that
they simulate the receive behavior of $i$ and $j$ in this round in a consistent manner.
To do so, in the first round of this group,
$i'$ (resp. $j'$) broadcasts if $i$ (resp. $j$) broadcasts in this simulated round.
If exactly one node broadcasts, wake-up is solved and we are done.
If neither node broadcasts, then both nodes receive and detect nothing (recall, we assume
${\cal A'}$ runs in a setting with a single channel and collision 
detection).
The two nodes can, at this point, simulate $i$ and $j$ also receiving and detecting nothing,
and skip ahead to the next simulated round and group.
The interesting case is if both nodes broadcast.
In this case, both nodes detect a collision in ${\cal A'}$.
To properly advance the simulation, $i'$ and $j'$ must decide
whether or not $i$ and $j$ should also detect a collision---they should
if $i$ and $j$ choose the same channel in this simulated round,
but should not if $i$ and $j$ choose different channels in this simulated round.

Let $c_i$ be the channel chosen by $i$ and $c_j$ the channel chosen by $j$
in this simulated round. 
Let $b_i$ be the binary representation of $c_i$, and $b_j$ the binary representation
of $c_j$. Notice, $|b_i| = |b_j| = \lceil \log{\chan} \rceil$.
To determine if $c_i$ and $c_j$ are equivalent,
$i'$ and $j'$ spend one round checking each bit in $c_i$ and $c_j$.
In each such round $k$, $i'$ broadcasts if bit $k$ of $b_i$ is $1$,
 and $j'$ broadcasts if bit $k$ of $b_j$ is $1$.
If $b_i \neq b_j$, then during one of these rounds exactly one node will broadcast,
solving wake-up.
If the nodes make it through all $\lceil \log{\chan} \rceil$ bits without solving wake-up,
then it follows that $c_i = c_j$.
This knowledge allows $i'$ and $j'$ to conclude their current simulated
round by simulating $i$ and $j$ both detecting a collision.

It is clear to see that $i'$ and $j'$ running ${\cal A'}$ on a single
channel with collision detection
correctly simulate $i$ and $j$ running ${\cal A}$ on $\chan$
channels with collision detection.
If ${\cal A}$ solves restricted wake-up in $f(n,\chan)$
rounds with high probability, then ${\cal A'}$ solves
restricted wake-up in $\leq g(n,\chan) = f(n,\chan)\cdot(\lceil \log{\chan} \rceil + 1)$ rounds,
with high probability.
Recall, however,
that we assumed
 $f(n,\chan) \in o(\log{n}/\log{\chan})$.
 It follows that $g(n,\chan) \in o(\log{n})$.
 We now have an algorithm, ${\cal A'}$,
 that solves restricted wake-up in a single channel with collision detection in $o(\log{n})$
 rounds.
 We can therefore directly apply the simulation strategy argument from
 Lemma~\ref{lem:multi:2} to ${\cal A'}$
 to prove the existence of a player that solves the restricted $k$-hitting
 game in $o(\log{k})$ time, also with high probability.
 This contradicts Theorem~\ref{thm:hitting:restricted}.

We now consider the $\log\log{n}$ term
of our lower bound. This term dominates our bound when $\chan$
is large (i.e., $\chan > n/\log{n})$.
To begin, as before, 
assume that some wake-up algorithm ${\cal A}$ solves
the problem in this setting in $o(\log\log{n})$ rounds for 
some $\chan$.
Also as before, we will confine our attention to restricted wake-up,
and construct a single channel algorithm ${\cal A'}$
that has two nodes $i'$ and $j'$ simulate two nodes
$i$ and $j$ running ${\cal A}$ with $\chan$ channels.
In the previous argument, the difficult case in this simulation is when $i$ and $j$
both broadcast.
To simulate collision detection properly, $i'$ and $j'$
must decide whether or not $i$ and $j$ chose the same channel.
To resolve this question, we can no longer directly apply the bit-by-bit approach
used above,
because if $\chan$ is large this would take too many rounds.

We must instead use the same type
of {\em simulation tree} argument
introduced in the proof of Theorem~\ref{thm:wakeup:cd:expectation}.
In particular, for each round, there are a constant number
of possible receive behaviors for $i$ and $j$.
We can consider a simulation tree that explores all possible
such behaviors at the cost of exponentiating the runtime.
If the path in the tree matching the correct receive behavior
solves wake-up, then the simulation will eventually 
test this path after no more than $2^{h+1}$ guesses,
where $h$ is the height.
Because we assume ${\cal A}$ solves the problem in $o(\log\log{n})$ rounds,
the simulation strategy used by ${\cal A'}$ solves it
in $o(\log{n})$ rounds.
As before, we now have a solution to restricted wake-up 
that solves the problem in a single channel with collision 
detection in $o(\log{n})$ rounds.
We obtain our contradiction with Theorem~\ref{thm:hitting:restricted}
in the same manner as with the first term.
\end{proof}

 %%%%%%%%%%%%%%%%%%%%%%%%%%%%%%%%%%%%%%%%%%%%%%%%%%%%%%%%
%%%%%%%%%%%%%%%%%%%%%%%%%%%%%%%%%%%%%%%%%%%%%%%%%%%%%%%%
%%%%%%%%%%%%%%%%%%%%%%%%%%%%%%%%%%%%%%%%%%%%%%%%%%%%%%%%

\section{Lower Bound for Global Broadcast}
 \label{sec:bcast}
 
 We now turn our attention to proving a lower bound for global broadcast.
 The tight bound for this problem is $\Theta(D\log{(n/D)} + \log^2{n})$ rounds
 for a connected multihop network of size $n$ with diameter $D$.
The lower bound holds for expected time solutions and the matching upper
bounds hold with high probability~\cite{baryehuda:1992,kowalski:2005,czumaj:2006}.
The $\log^2{n}$ term was established in~\cite{alon:1991}, where it was shown
to hold even for centralized algorithms,  
and the $D\log{(n/D)}$ term was later proved by Kushilevitz and Mansour~\cite{kushilevitz:1998}.
Below,
 we apply our new technique
 to reprove (and significantly simplify) the
  $\Omega(D\log{(n/D)})$ lower bound for expected time solutions to global broadcast.
(We do not also reprove the $\Omega(\log^2{n})$ term
 because this bound is proved using the same combinatorial result from~\cite{alon:1991} 
  that provides the mathematical foundation for our technique.
  To reprove the result of~\cite{alon:1991} using~\cite{alon:1991} is needlessly circular.)
 Perhaps surprisingly, we show that this bound holds even if we allow multiple
 channels and collision detection, both of which are assumptions that break the original 
 lower bound from~\cite{kushilevitz:1998}.
 Notice, this indicates a interesting split with the wake-up problem for which
 these assumptions {\em improve} the achievable time complexity.
 
 %%%%
 \iffalse
 .\footnote{Recall that in our definition of global broadcast, only the nodes
 that have been activated by receiving the message can use their collision detectors.
 The variant of the problem where collision detection can activate nodes that have not yet
 received the message is quite different; c.f.,~\cite{ghaffari:2013b}.}
 This represents a {\em significant split} between global broadcast and wake-up.
 As shown in the preceding sections,
  multiple channels and collision detection both enable faster solutions to wake-up (i.e., these
  assumptions reduce the tight bounds for the problem).
  We show below that these assumptions do not help global broadcast.
  The intuition for this split concerns the multihop nature of the broadcast problem.
  Imagine a clique of active nodes with the broadcast
  message, arranged in a  multihop network such that 
  some unknown subset of the clique is connected to the next hop.
  The presence of multiple channels or collision
  detection can help these active nodes reduce contention {\em amongst themselves},
  but it does not help them discover which of them are connected to the next hop.
  
  Notice, in the following we do not also reprove the $\Omega(\log^2{n})$ term.
  This is because this bound is proved using the same combinatorial result from~\cite{alon:1991} 
  that provides the mathematical foundation for our technique.
  To reprove the result of~\cite{alon:1991} using~\cite{alon:1991} is needlessly circular.
  \fi
  %%%%
   \vfill \eject
   
 \begin{theorem}
Let ${\cal A}$ be an algorithm that solves
global broadcast in $f(n,\chan,D)$ rounds, in expectation,
in the radio network model with collision detection,
$\chan\geq 1$ channels, and a network topology with diameter $D$.
It follows that for every $\chan,D \geq 1$, $f(n,\chan,D) \in \Omega(D\log{(n/D)})$.
\label{thm:bcast}
\end{theorem}
\begin{proof}
To account for the multihop nature of the problem,
we introduce a natural generalization of the hitting games introduced in Section~\ref{sec:hitting}.
In more detail, the $(k,k')$-{\em multi-hitting game}, for $1 \leq k' \leq k$, is a variation
of $k$-hitting game in which we run $k'$ consecutive instances of the
($\lfloor k/k' \rfloor$)-hitting game,
requiring the player to win instance $i \in \{1,...,k' -1\}$ before proceeding to instance $i+1$.
There are two technical points in this definition
 that aid the below argument: assume that the referee selects all $k'$ targets at the beginning of the game,
and assume that the referee reveals to the player the target for instance $i$ at the end of the round in which the player wins that instance.

It is straightforward to use the existing bounds from Section~\ref{sec:hitting}
to bound this generalization. 
Consider a particular instance of this game for a particular player ${\cal P}$.
 let $X_i$, for $i\in \{1,...,k'\}$, be the time required
to win trial $i$ of the game, and let $Y = X_1 + X_2 + ... + X_{k'}$ be the time required to win the
full multi-set game.
Let $\mathbb{E}_{\cal P}[X_i]$, for each relevant $i$,
be the expected time for ${\cal P}$ to win trial $i$,
and $\mathbb{E}_{\cal P}[Y]$ be the expected time to win the full multi-set game.
Imagine that we apply the referee target selection strategy from Theorem~\ref{thm:hitting:expectation} (our lower bound
on the expected time for $k$-hitting) for each trial in the multi-set hitting game, using independent randomness
to make each selection.
For each $X_i$, $i\in [k']$, it follows from Theorem~\ref{thm:hitting:expectation} that regardless of ${\cal P}$'s definition,
 $\mathbb{E}_{\cal P}[X_i] \geq \log{\lfloor k/k' \rfloor}$.
We can now lower bound $\mathbb{E}_{\cal P}[Y]$ by leveraging linearity of expectation:
%
 %\begin{equation}
$\mathbb{E}_{\cal P}[Y] = \mathbb{E}_{\cal P} \big[  \sum_i^{k'} X_i \big] = \sum_i^{k'} \mathbb{E}_{\cal P}[X_i] \geq k'\log{\lfloor k/k'\rfloor}.$
 %\end{equation}

 Having bounded $\mathbb{E}_{\cal P}[Y] $,
  we can proceed to our main argument.
 Fix some ${\cal A}$ that solves global broadcast in $f(n,\chan,D)$ rounds, in expectation, in networks
 of size $n$ with $\chan$ channels and diameter $D$.
 Fix any valid values for $\chan$ and $D$.
 We will now prove the existence of a network of diameter $D$ in 
 which ${\cal A}$ requires $D\log{(n/D)}$ rounds to solve broadcast in expectation,
 even when
 provided $\chan$ channels and collision detection. 
To do so, we deploy a variant of the simulation strategy introduced in Section~\ref{sec:sim},
that will be used by a multi-hitting game player ${\cal P}_{{\cal A,C},D}$
to play the hitting game by simulating ${\cal A}$ running in specific diameter $D$ network with $\chan$ channels 
and collision detection.
In particular,
our player simulates ${\cal A}$ on a 
 network consisting of $D+1$ layers, $L_1,L_2,...,L_D,L_{D+1}$,
 where the first $D$ layers each include $\lfloor n/D \rfloor$ nodes,
 and the last layer includes at least $1$ node (if $D$ divides $n$ evenly,
 then we can add an extra node to the system to populate $L_{D+1}$, without affecting
 the asymptotic bounds below; otherwise we add the leftover node(s) to this last layer).
  For the sake of construction, for each $L_i$, assign unique labels from $\{1,...,k\}$ to the nodes in $L_i$.
  Let $T_i$ be the target chosen by the referee for trial $i$ of the instance of the multi-hitting game being played by our player . 
  In our construction, we connect $L_i$ and $L_{i+1}$ by including an edge
  from every node in $L_i$ with a label corresponding to a value in $T_i$ to every node in $L_{i+1}$.
    Notice, the player simulating this network {\em does not know} these $T_i$ values in advance,
    and therefore does not know the full topology of the network on which it is simulating ${\cal A}$,
     but we will now
  show this does not matter as the simulation will remain consistent. Finally, the nodes within each layer are connected as a clique.
(Notice that this graph satisfies the unit disk graph property---strengthening our bound even beyond what
is stated in the above theorem to indicate it holds even if we restrict our attention to unit disk graphs: an easier
setting than general graphs.)
  
  The simulation begins with the player choosing some node in $L_1$ as the source.  
In each round $r$ of the simulation,
let $\hat i$ be the largest value of $i$ such that the nodes in $L_i$ are active (i.e., have the message).
Let $B^r_{\hat i}$ be the nodes in $L_{\hat i}$ that broadcast in $r$, if any.
If $\chan >1$, let $B^r_{\hat i}$ be the nodes in $L_{\hat i}$ that broadcast on the default
channel where inactive nodes listen (i.e., channel $1$).
The player uses $B^r_{\hat i}$ as its proposal in this round of the mutli-hitting game.
 The key insight of this reduction is that the player only needs to simulate
 communication between $L_{\hat i}$ and $L_{\hat i+1}$ if exactly one node connecting $L_{\hat i}$ to $L_{\hat i+1}$
 is in $B^r_{\hat i}$. When this occurs, the player will learn of this fact, because its corresponding
 guess in the hitting game will win this instance of the game (and once it wins instance $i$
 for the first time, it learns $T_i$, so it can, moving forward, successfully simulate all future communication between
 these two layers). 
 The player knows the full topology of all smaller layers, so it can always correctly simulate the behavior of nodes
 broadcasting in these layers as well.
 
 Collision detection and multiple channels break the original proof of~\cite{kushilevitz:1998} because their
 argument requires that nodes in the same layer receive silence in all rounds before
 they advance the message. If the active nodes in a layer had collision detection, for example, 
 they could quickly achieve some communication using collisions, at which point the argument of \cite{kushilevitz:1998} fails.
 Our argument can tolerate such intra-layer communication as it focuses only on the externally observable
 (i.e., broadcast) behavior of the layer.
 
 We conclude by noting that the player using this strategy will win the multi-hitting game
 when the message arrives at $L_{D+1}$. 
  By assumption, this occurs in expected time of $f(n,\chan,D)$ rounds
  By our above bound on $\mathbb{E}[Y]$,
  and the fact that $k$ corresponds to $n$ and $k'$ to $D$,
   it must follow that $f(n,\chan,D) \in \Omega(D\log{(n/D)})$, as needed.
  \end{proof}

\section*{Acknowledgments}

The author acknowledges Mohsen Ghaffari for his many helpful conversations
regarding the combinatorial results from~\cite{alon:1991,ghaffari:2013} that form
the core of the hitting game lower bounds.
The author also acknowledges Sebastian Daum and Fabian Kuhn who co-authored with myself
the first paper to apply this general strategy to a classical radio network lower bound~\cite{daum:2012}.

 %%%%%%%%%%%%%%%%%%%%%%%%%%%%%%%%%%%%%%%%%%%%%%%%%%%%%%%%
%%%%%%%%%%%%%%%%%%%%%%%%%%%%%%%%%%%%%%%%%%%%%%%%%%%%%%%%
%%%%%%%%%%%%%%%%%%%%%%%%%%%%%%%%%%%%%%%%%%%%%%%%%%%%%%%%
%
%%%%%%%%%%%%%%%%%%%%%%%%%%%%%%%%%%%%%%%%%%%%%%%%%%%%%%%%
%%%%%%%%%%%%%%%%%%%%%%%%%%%%%%%%%%%%%%%%%%%%%%%%%%%%%%%%
%%%%%%%%%%%%%%%%%%%%%%%%%%%%%%%%%%%%%%%%%%%%%%%%%%%%%%%%

%\section{Conclusion}

 %%%%%%%%%%%%%%%%%%%%%%%%%%%%%%%%%%%%%%%%%%%%%%%%%%%%%%%%
%%%%%%%%%%%%%%%%%%%%%%%%%%%%%%%%%%%%%%%%%%%%%%%%%%%%%%%%
%%%%%%%%%%%%%%%%%%%%%%%%%%%%%%%%%%%%%%%%%%%%%%%%%%%%%%%%

% \newpage
 
%-----------------------------------------------------------------------------------
\bibliographystyle{abbrv}
\bibliography{../../nsf}
%-----------------------------------------------------------------------------------

\end{document}